\documentclass[11pt, letterpaper,  onecolumn]{IEEEtran}

\IEEEoverridecommandlockouts
\usepackage[margin=0.70 in]{geometry}
\usepackage{amsmath,amsfonts,amsthm,enumerate}
\usepackage{amssymb}
\usepackage{graphicx}
\usepackage{cite,bm}
\usepackage{color,epsfig}
\usepackage{authblk}
\usepackage[center]{caption}

\usepackage{setspace}

\theoremstyle{plain}
\newtheorem{mythm}{Theorem}
\newtheorem{mylemma}{Lemma}
\newtheorem{mycorollary}{Corollary}

\theoremstyle{definition}
\newtheorem{mydef}{Definition}
\newtheorem{myexample}{Example}

\theoremstyle{remark}
\newtheorem{myrm}{Remark}

\DeclareMathOperator*{\argmin}{arg\,min}

\usepackage[ colorlinks = true,
             linkcolor = blue,
             urlcolor  = blue,
             citecolor = red,
             anchorcolor = green,
]{hyperref}

\allowdisplaybreaks[1]

\begin{document}
%\flushbottom
\title{Second-Order Coding Rates for Conditional Rate-Distortion} 
%\thanks{The work of M. Motani and S.-Q. Le was supported in part by the National University of Singapore under Research Grant WBS R-263-000-579-112. E-mail:
%\{le.sy.quoc, motani\}@nus.edu.sg}\thanks{The work of V. Y. F. Tan was supported by  an NUS startup grant WBS R-263-000-A98-750. Email:vtan@nus.edu.sg} 

%\author{Sy-Quoc Le}
%\author{Vincent Y. F. Tan}
%\author{Mehul Motani}
\author{Sy-Quoc Le,$\,\,\,$ Vincent Y.~F.\ Tan,$\,\,\,$ Mehul Motani \thanks{The authors are   with the Department of Electrical and Computer Engineering (ECE), National University of Singapore (NUS). V.~Y.~F.\ Tan is also with the Department of Mathematics, NUS.  The authors' emails are \url{le.sy.quoc@nus.edu.sg},  \url{vtan@nus.edu.sg} and  \url{motani@nus.edu.sg}. } } 
 
%\affil{\normalsize Department of Electrical and Computer Engineering,      National University of Singapore}
\maketitle
\begin{abstract}
This paper characterizes the second-order coding rates for lossy source coding with side information available at both the encoder and the decoder. We first provide non-asymptotic bounds for this problem and then specialize the non-asymptotic bounds for three different scenarios: discrete memoryless sources, Gaussian sources, and Markov sources. We obtain the second-order coding rates for these settings. It is interesting to observe that the second-order coding rate for Gaussian source coding with Gaussian side information available at both the encoder and the decoder is the same as that for Gaussian source coding without side information.  Furthermore, regardless of the variance of the side information, the dispersion is $1/2$ nats squared per source symbol.
\end{abstract}

\section{Introduction}
In almost lossless source coding, the Shannon entropy of a source is, on average, the minimum number of bits required to represent a given source \cite{Shannon48}. In lossy source coding, the rate-distortion function (which, in this paper, is more specifically called  the rate-distortion function without side information) plays the role of the Shannon entropy \cite{Shannon59S}. The rate-distortion function without side information is the minimum number of bits per symbol required to reconstruct a given source with the probability of excess distortion being asymptotically small, or with an average distortion that does not exceed a specified upper bound. 

The class of  source coding problems with side information is important as it can model many practical problems. Consider a scenario when a source wants to  transmit a high-resolution image to a receiver who happens to have a low-resolution version of the same image. In another example, the source may be a piece of music contaminated by a background noise source  and the intended receiver has already had observations of the background noise.  The rate-distortion problem without side information  can be extended to the case when the side information is available at both the encoder and the decoder \cite{Berger71, Gray72}, only causally available at the decoder \cite{Weissman06}, or non-causally available at the decoder (i.e., Wyner-Ziv problem) \cite{WZ76}. 
%In the almost lossless source coding problem, the rate-distortion function for the case when side information is only available at the encoder is same as that for the case without side information \cite{elgamal}. 
The rate-distortion function for  stationary-ergodic sources with side information was found in \cite{LG}. The rate-distortion function for mixed types of side information (i.e., a mixture of some side information known at both the encoder and the decoder and some known only at the decoder) was evaluated in \cite{FE06}. For memoryless sources, delayed side information at the decoder does not improve the rate-distortion function. However, this is not the case for sources with memory \cite{SP13}. The authors of \cite{LZZ00} considered source coding with side information, and  with distortion measures as functions of side information. 

All the results shown above hold provided the blocklength, i.e., the number of source symbols, is allowed to grow without bound. However, some applications are required to operate with short blocklengths due to delay or complexity constraints at the destination. Thus, it is of high interest to characterize the finite blockength rate-distortion function, i.e., the minimum number of bits per symbol that is required to reconstruct a source at a given fixed blocklength. This is, in general, a difficult task, and thus, we focus on approximating this quantity. 
 \subsection{Related Works}

Strassen \cite{Strassen62} obtained the second-order coding rate for  almost lossless source coding without side information.  Recently, Hayashi \cite{Hayashi08} considered second-order coding rate for fixed-length source coding and showed that the outputs of fixed-length source codes are not uniformly distributed (debunking Han's folklore theorem \cite{Han_folklore} in the second-order sense).  Kostina and Verd\'u~\cite{KV12} and Ingber and Kochman~\cite{IK11} characterized the dispersion of lossy source coding problem without side information. When the source is stationary and memoryless, they showed that the finite blockength rate-distortion function without side information $R_{\mathsf{noSI}}(n,D,\epsilon)$ can be approximated as
\begin{align}
R_{\mathsf{noSI}}(n,D,\epsilon) = R_{\mathsf{noSI}} (D) + \frac{\sqrt{V_{\mathsf{noSI}}(D)}}{n} Q^{-1}(\epsilon) + O \left(\frac{\log n}{n} \right),
\end{align}
where $R_{\mathsf{noSI}}(D)$ is the rate-distortion function without side information,  $V_{\mathsf{noSI}}(D)$ is the \textit{dispersion} that characterizes the convergence rate to the Shannon limit  $R_{\mathsf{noSI}}(D)$, $n$ is the blocklength, $D$ is the excess distortion threshold, and $\epsilon$ is the upper bound on the probability that the distortion exceeds $D$. The rate-distortion problem may also be studied from the moderate deviations perspective \cite{Tan12} and the fundamental limit there is also dependent on $V_{\mathsf{noSI}}(D) $. Achievable second-order coding rates for the Wyner-Ahlswede-Korner problem of almost-lossless source coding with rate-limited side-information, the Wyner-Ziv problem of lossy source coding with side-information at the decoder and the Gelfand-Pinsker problem of channel coding with non-causal state information available at the decoder were established in \cite{WKT13}. The paper \cite{KontoVerdu12} studied  second-order coding rates for the fixed-to-variable lossless compression. For other related works in the study of fixed error asymptotics, the reader is referred to~\cite{Tan_mono}.

 \subsection{Main Contributions}
This paper focuses on the analysis  and approximation of the  finite blockength rate-distortion function for source coding with side information available at both the encoder and the decoder. The contributions of this paper are stated below.
\begin{itemize}
\item A non-asymptotic achievability bound is established for the problem of lossy source coding with side information available at both the encoder and the decoder.
\item We establish the second-order coding rate for the discrete memoryless source with a side information variable taking values in a finite alphabet. As a corollary, we obtain the second-order coding rate for the case when the source alphabet, the reconstruction alphabet and the side information alphabet are finite and the distortion measure is the Hamming distance.
\item We establish the second-order coding rate for Gaussian source with Gaussian side information and the squared-error distortion measure. Somewhat interestingly, the dispersion does not depend on the variance of the side-information and is $1/2$ squared nats per source symbol. 
\item When the source has memory, we establish the second-order coding rate for the case where the  sequence of source  and side information variables jointly forms a time-homogeneous Markov chain.  
\end{itemize}
 
 \subsection{Paper Outline}
 The paper is organized as follows. In section \ref{C5sec:problem} we formulate the problem, and define important concepts which are used throughout the paper. In section \ref{C5sec:nonasymptotic} we present non-asymptotic bounds for the source coding problems with side information available at both the encoder and the decoder. These so-called one-shot bounds hold for any blocklength. Based on the bounds established in section \ref{C5sec:nonasymptotic}, we establish the second-order coding rates for the discrete memoryless source, the Gaussian source and the Markov source in sections \ref{C5sec:DM}, \ref{C5sec:G} and \ref{C5sec:M} respectively. Technical  proofs are presented in section \ref{C5sec:app}.

 \section{Problem formulation and definitions} \label{C5sec:problem}
 Let $\mathcal{X}$ be the source alphabet, let $\mathcal{Y}$ be the reproduction alphabet, and let $\mathcal{S}$ be the side information alphabet. The random variables $X,Y$ and $S$ follow the distribution
\begin{align}
P_{YXS}(yxs) = P_{Y|XS}(y|xs) P_{X|S} (x|s) P_{S}(s). \label{C5E1}
\end{align}
We use a single-letter fidelity criterion to measure the distortion between the source sequence $x^n$ and the reproducing sequence $y^n$, i.e.,
 \begin{align}
 d(x^n, y^n) = \frac{1}{n} \sum_{i=1}^n d(x_i, y_i),
 \end{align}
 where $d:\mathcal{X}^n \times \mathcal{Y}^n \to \mathbb{R}_+$, for $n \in \mathbb{N}$, is a bounded real-valued non-negative distortion function. 
 \begin{mydef}
 
\begin{figure}
\centering
\includegraphics[width=0.6\textwidth]{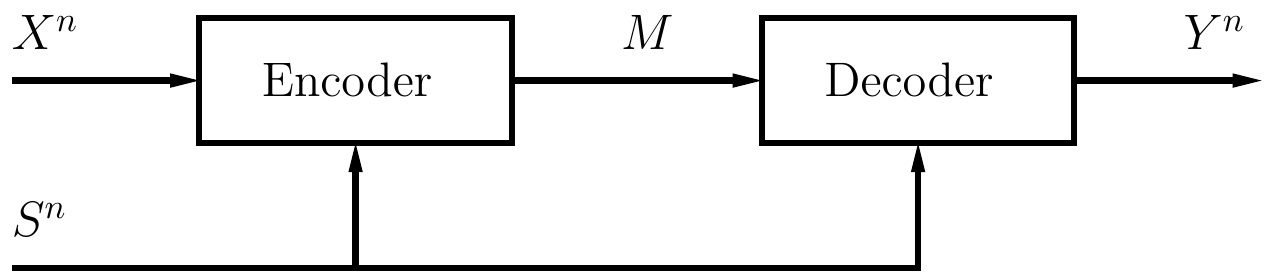}
\caption{Source coding with side information}
\label{fig:C5SCSI}
\end{figure}

 An $(M_n,n,D,\epsilon_n)$-code for the source coding system with side information (see Figure \ref{fig:C5SCSI}) consists of an encoding function
 \begin{align}
 \phi_n: \mathcal{X}^n \times \mathcal{S}^n \to \mathcal{M}_n \triangleq \{1,2,\ldots, M_n\},
 \end{align}
 and a decoding function
  \begin{align}
 \psi_n: \mathcal{M}_n \times \mathcal{S}^n \to \mathcal{Y}^n,
 \end{align}
 such that the probability of excess distortion satisfies
 \begin{align}
 \Pr\{d[X^n, \psi_n(\phi_n(X^n,S^n),S^n)] > D\} \leq \epsilon_n.
 \end{align}
 \end{mydef}
 An $(M_n,n,D,\epsilon_n)$-code, which is defined as shown above, is called a \textit{$D$-semifaithful} code in the rate-distortion literature \cite{YS93, ZYW}.
 
 \begin{mydef}
 A rate $R$ is defined to be \textit{$(\epsilon,D)$-achievable} if there exists a sequence of $(M_n,n,D,\epsilon_n)$-codes satisfying
 \begin{align}
 \limsup_{n \to \infty} \frac{1}{n} \log M_n &\leq R, \\
 \limsup_{n \to \infty} \epsilon_n  &\leq \epsilon.
 \end{align}
 \end{mydef}
 
 In contrast to the above definition, the following definition is non-asymptotic. 
 \begin{mydef}
 A rate $R$ is defined to be \textit{$(\epsilon,D,n)$-achievable} if there exists a $(\lfloor \exp(nR)\rfloor,n,D,\epsilon_n)$-code.
 The $(\epsilon,D,n)$ finite blockength rate-distortion function $R(\epsilon,D,n)$ is defined as the infimum of the set of all $(\epsilon,D,n)$-achievable rates.
 \end{mydef}
%\begin{myrm}
%$R(\epsilon,D)$ can be thought of as a first-order rate-distortion function.
%\end{myrm}
 
 The following definition defines the quantity of interest in this paper.
 \begin{mydef}
 A number $L \in \mathbb{R}$ is defined to be \textit{second-order $(\epsilon,D,\kappa)$-achievable} if there exists a sequence of   $(M_n,n,D,\epsilon_n)$-codes satisfying
 \begin{align}
 \limsup_{n \to \infty} \frac{1}{\sqrt{n}} (\log M_n - n \kappa) &\leq L,\\
 \limsup_{n \to \infty} \epsilon_n  &\leq \epsilon.
 \end{align}
 The \textit{$(\epsilon,D,\kappa)$ second-order rate-distortion function} $L^*(\epsilon,D,\kappa)$ is defined as the infimum of the set of all second-order $(\epsilon,D,\kappa)$-achievable rates.
 \end{mydef}
 
The aim of this paper is to characterize the $(\epsilon,D,\kappa)$ second-order rate-distortion function $L^*(\epsilon,D,\kappa)$ for source coding with side information available at both the encoder and the decoder.

Before presenting the main result, we state some definitions that will be used throughout this paper.

\begin{mydef} \label{C5Erdf}
Fix the distribution of $XS$ as $P_{XS}$. Define the rate-distortion function with side information as
\begin{align}
R(X;D|S) = \min_{P_{Y|XS}} I(X;Y|S), \label{C5Erdf1}
\end{align} 
where the minimum is taken over the set of all marginal conditional distributions $P_{Y|XS}$ satisfying 
\begin{align}
P_{Y|XS}(y|xs) &\geq 0 \qquad \mbox{ for all $(y,x,s)$}, \\
\sum_{y \in \mathcal{Y}} P_{Y|XS}(y|xs) &= 1, \\
\sum_{s \in \mathcal{S},x \in \mathcal{X}, y \in \mathcal{Y}} P_{Y|XS}(y|xs) P_{X|S}(x|s) P_S(s) d(x,y) &\leq D.
\end{align}
To make  the dependence on the distribution $P_{XS}$ explicit, we sometimes also denote $R(X;D|S)$ as $R(P_{X|S},D|P_S)$. Assume the distribution that achieves the minimum in (\ref{C5Erdf1}) is unique. When there is no side information, i.e., $S= \emptyset$, we recover the rate-distortion function without side information denoted as $R(X;D)$ or $R(P_{X},D)$.
\end{mydef}

When the excess distortion criterion is employed, we have the following first-order result for the source coding problem with side information \cite{Berger71} (i.e., the conditional rate-distortion problem \cite{Gray72}),
\begin{align}
\lim_{\epsilon \to 0} \liminf_{n \to \infty} R(\epsilon,D,n) = R(X;D|S).
\end{align} 
%This result can be shown to hold also in the case when excess distortion constraint is used.

%\begin{mydef}
%Given that $(Y,X,S)$ are distributed according to distribution in (\ref{C5E1}), define the rate-distortion function with deterministic side information as
%\begin{align}
%R(X;D|S=s) \triangleq \min_{P_{Y|XS}} I(X;Y|S=s),
%\end{align} 
%where the minimum is taken over the set of all marginal conditional distributions $P_{Y|XS}$ satisfying 
%\begin{align}
%P_{Y|XS}(y|xs) &\geq 0 \mbox{ for all $(y,x,s)$}, \\
%\sum_{y \in \mathcal{Y}} P_{Y|XS}(y|xs) &= 1, \\
%\sum_{x \in \mathcal{X}, y \in \mathcal{Y}} P_{Y|XS}(y|xs) P_{X|S}(x|s) d(x,y) &\leq D.
%\end{align}
%For convenience, we sometimes also denote $R(Y;X|S=s)$ as $R(P_{X|S=s}(\cdot|s),D)$.
%\end{mydef}
%\begin{myrm}
%Note that $R(X;D|S=s)$ can be thought as the rate-distortion function for the system when the side information is fixed.
%\end{myrm}

In order to characterize the second-order rate-distortion function, we state the following definitions.  The notion of information densities will play an important role in characterizing the second-order rate-distortion function. In fact, in order to deal with the constraints inherent in the rate-distortion problem, the concept of $D$-tilted information densities, which was introduced in  \cite{KV12CISS}, is useful.
\begin{mydef}
Define the conditional information densities as follows:
\begin{align}
i_{X;Y|S}(x;y|s) &\triangleq \log \frac{P_{XY|S}(xy|s)}{P_{Y|S}(y|s) P_{X|S}(x|s)} ,\quad\mbox{and}\\
i_{X|S}(x|s)     &\triangleq i_{X;X|S}(x;x|s). \label{eqn:self_info}
\end{align}
Note that $i_{X|S}$ is also known as the conditional self-information. 
\end{mydef}

\begin{mydef}
Define the conditional $D$-tilted information density as follows:
\begin{align}
j_{X|S}(x,D|s) \triangleq \log \frac{1}{\mathbb{E}[\exp\{\lambda^*D - \lambda^* d(x,Y^*)\}|S=s]}
\end{align}
where $P_{Y^*|XS}$ is the  distribution that achieves the minimum in (\ref{C5Erdf}), the expectation is taken with respect to the induced output distribution $P_{Y^*|S}(y|s) = \sum_x P_{Y^*|XS}(y|x,s)P_{X|S}(x|s)$, and $\lambda^*$ is defined as
\begin{align}
\lambda^* \triangleq \frac{dR(P_{X|S},D|P_S)}{dD}.
\end{align}
\end{mydef}
\begin{myrm}
In this definition, the conditional $D$-tilted information density has a built-in feature which takes  the distortion constraint  into consideration. %This is the advantage of the conditional $D$-tilted information density over the conditional non-tilted counterpart. 
\end{myrm}

The conditional $D$-tilted information density $j_{X|S}(x,D|s)$ has some important properties which can be  found in \cite{KV12CISS}. We review them here.
\begin{mylemma} \label{C5Ldtid}
The conditional $D$-tilted information density $j_{X|S}(x,D|s)$ has  the following properties.
\begin{enumerate}
\item $j_{X|S}(x,D|s)   = i_{X;Y*|S}(x;y|s) + \lambda^*d(x,y) - \lambda^*D.$ 
\item $R(X;D|S) = \mathbb{E}[j_{X|S}(X,D|S)].$ 
\item For any $P_{Y|S}$ where $X \to S \to Y$, we have  $\mathbb{E}[\exp\{\lambda^*d - \lambda^*d(X,Y) + j_{X|S}(X,D|S)\}] \leq 1. $
\end{enumerate}

\end{mylemma}

In the achievability proof of the conditional rate-distortion problem, the following concept is important.
\begin{mydef}
Given a source sequence $x^n \in \mathcal{X}^n$, define the $D$-ball $B_{D}(x^n)$  around this sequence as 
\begin{align}
B_D{(x^n)} \triangleq \{y^n \in \mathcal{Y}^n | d(x^n, y^n) \leq D \}.
\end{align}
\end{mydef}

%\begin{mydef}
%A set $B \in \mathcal{Y}^n$ is said to \textit{D-cover} a set $A \in \mathcal{X}^n$ if
%\begin{align}
%\max_{x^n \in A, y^n \in B} d(x^n, y^n) \leq D. 
%\end{align}
%\end{mydef}
%\begin{myrm}
%Note that when $B$ D-covers $A$, we can easily construct a $(|B|,n,D,0)$-code whose source sequences are elements of $A$ only. With some abuse of language, we also say the code $B$ D-covers the set $A$.
%\end{myrm}

%\begin{mydef}
%For a given conditional type $T_{x^n|s^n}$, a given positive constant D, and  a set $B$ in $\mathcal{Y}^n$, a subset $U(B,D,T_{x^n|s^n},s^n)$ of conditional type class $\mathcal{T}^n(T_{x^n|s^n})$ is defined as 
%\begin{align}
%U(B,D,T_{x^n|s^n},s^n) \triangleq \{x^n| x^n \in T_{x^n|s^n}, \min_{y^n \in B} d(x^n, y^n) >D \}.
%\end{align}
%In other words, $U(B,D,T_{x^n|s^n},s^n)$ is a subset, that is not D-covered by the set $B$, of the conditional type class $\mathcal{T}^n(T_{x^n|s^n})$.
%\end{mydef}

The following  is the  cumulative distribution function of a standard Gaussian distribution
\begin{equation}
\Phi(t)\triangleq  \int_{-\infty}^t  \frac{1}{\sqrt{2\pi}}\exp(-u^2/2)\,\mathrm{d} u. 
\end{equation}
The complementary cumulative distribution function is $Q(t) \triangleq 1-\Phi(t)$. Since these functions are monotonic, they admit inverses, which we will denote as $\Phi^{-1}$ and $Q^{-1}$.

\section{Non-Asymptotic Bounds} \label{C5sec:nonasymptotic}
In this section, we first present a non-asymptotic achievability bound.
\begin{mylemma}[Achievability] \label{C5Lrcb}
For every  $P_{\bar{Y}^n|S^n} $, there exists an $(M_n, D, n, \epsilon_n)$-code such that
\begin{align}
\epsilon_n \leq \mathbb{E} \{\mathbb{E}[(1 - P_{\bar{Y}^n|S^n} (B_D(x^n) |S^n)^M] \} \label{C5Ercb}
\end{align}
where we have
\begin{align}
P_{\bar{Y}^n X^n S^n} = P_{\bar{Y}^n|S^n} P_{X^n|S^n} P_{S^n}. \label{C5Ercb1}
\end{align}
\end{mylemma}
\begin{proof}
Given each side information sequence $S^n = s^n$, we construct a reconstruction codebook $\mathcal{C}(s^n)$, which consists of $M$ random reconstruction sequences $\{Y^n(m,s^n)\}_{m=1}^M$. Each of the sequence $Y^n(m,s^n)$, for $m \in \mathcal{M} \triangleq \{1,2,\ldots,M \}$, is generated independently  according to an arbitrary distribution $P_{\bar{Y}^n|S^n= s^n}$, which satisfies equation (\ref{C5Ercb1}).
Choose a sub-code $(\phi_n, \psi_n)$, the encoder and decoder of which are defined as
\begin{align}
\phi_n(x^n,s^n) &= \argmin_{m \in \mathcal{M}} d(x^n, Y^n(m,s^n)), \\
\psi_n (m,s^n)  &= Y^n(m,s^n).
\end{align}
The average probability of error of this sub-code is given by
\begin{align}
\bar{\epsilon} (s^n) &=  \mathbb{E}[1\{ \min_{m \in \mathcal{M}} d(X^n, Y^n(m,s^n)) >D\}|S^n = s^n] \\
                     &=  \mathbb{E}\left[ \prod_{m=1}^M 1\{ d(X^n, Y^n(m,s^n)) >D\}|S^n = s^n \right]  \\
                     &=  \mathbb{E}\left[ \mathbb{E} \left[\prod_{m=1}^M 1\{ d(X^n, Y^n(m,s^n)) >D\}|X^n\right]\bigg|S^n = s^n\right] \\
                     &=   \mathbb{E}\left[ \prod_{m=1}^M \mathbb{E} [1\{ d(X^n, \bar{Y}^n) >D\}|X^n]|S^n = s^n\right] \label{C5Ercb2} \\
                     &=   \mathbb{E}[(1 - P_{\bar{Y}^n|S^n} (B_D(X^n))|S^n=s^n)^M]
\end{align}
where equation (\ref{C5Ercb2}) follows from the independence of reconstruction sequences.

Taking the average over all sub-codes, we have the average probability of error is
\begin{align}
\bar{\epsilon} &= \sum_{s^n \in \mathcal{S}^n} P_{S^n} (s^n) \bar{\epsilon} (s^n) \\
               &= \mathbb{E}\{\mathbb{E}[(1 - P_{\bar{Y}^n|S^n} (B_D(X^n))|S^n)^M]\}.
\end{align}
By the random coding argument, there exists an $(M_n, D, n, \epsilon_n)$-code such that
\begin{align}
\epsilon_n \leq \mathbb{E} \{\mathbb{E}[(1 - P_{\bar{Y}^n|S^n} (B_D(X^n))|S^n)^M] \}.
\end{align}
This concludes the proof.
\end{proof}

Next, we relax the bound in Lemma \ref{C5Lrcb} to obtain the following lemma, which turns out to be more amenable to asymptotic evaluations.
\begin{mylemma} \label{C5Lforward}
For any $\gamma_n, \beta_n$, and $\delta_n$, there exists an $(M_n, D,n,\epsilon_n)$-code such that
\begin{align}
\epsilon_n &\leq \Pr [j_{X^n|S^n}(X^n,D|S^n) > \log \gamma_n - \log \beta_n - \lambda_n^* \delta_n] \notag\\
           &\quad + \mathbb{E}[\mathbb{E}[|1 - \beta_n \Pr[D -\delta_n \leq d(X^n,Y^{n*}) \leq D|X^n]|^+|S^n]] \notag \\
           &\quad + e^{-\frac{M}{\gamma_n}} \mathbb{E}\{\mathbb{E}[\min(1,\gamma_n \exp (-j_{X^n|S^n} (X^n,D|S^n)))|S^n]\},
\end{align}
where $P_{Y^{*}|X S}$ achieves the minimum in (\ref{C5Erdf}), and $P_{Y^{n*}|X^n S^n}$ is the $n$-th order product distribution of $P_{Y^{*}|X S}$.
\end{mylemma}
This lemma is proved in section \ref{C5ProofLforward}.

The following lemma, which plays an important part in the converse, was derived in \cite{KV12CISS}.
\begin{mylemma} \label{C5Lconverse}
Any $(M_n,n,D,\epsilon_n)$-code for the lossy source coding system with side information satisfies
\begin{align}
\epsilon_n \geq \sup_{\gamma >0} \{ \Pr[j_{X^n|S^n}(X^n,D|S^n) \geq \log M_n + \gamma] - \exp(-\gamma) \}.
\end{align}
\end{mylemma}

\section{Discrete memoryless source with i.i.d. side information} \label{C5sec:DM}
In this section, we consider the discrete memoryless source. Assume that the source alphabet $\mathcal{X}$,  the  reproduction alphabet $\mathcal{Y}$, and  the side information alphabet $\mathcal{S}$ are finite. 
 The source coding system is memoryless and stationary in the sense that 
 \begin{align}
 P_{X^nS^n}(x^ns^n) = \prod_{i=1}^n P_{ X S}(x_i s_i).
 \end{align}

Before presenting the main results of this section, we define an important quantity.
\begin{mydef}
Define the variance $V$ of the $D$-tilted information density $j_{X|S}(X,D|S)$ with respect to $P_{XS}$ as 
\begin{align}
V &\triangleq \mathsf{var}(j_{X|S}(X,D|S))  \label{eqn:var_1}\\
  &= \sum_{x \in \mathcal{X}, s \in \mathcal{S}} P_{XS}(xs)  [j_{X|S}(x,D|s)]^2 - [R(X;D|S)]^2.\label{eqn:var_2}
\end{align} 
\end{mydef}

Next, we present the first  main result of this paper.
\begin{mythm} \label{C5T1}
 The second-order rate-distortion function $L^*(\epsilon,D,R(X;D|S))$ for the discrete memoryless source coding with side information is given by
\begin{align}
L^*(\epsilon,D,R(X;D|S)) = \sqrt{V} Q^{-1}(\epsilon).
\end{align}
\end{mythm}

Let us mention that the \textit{dispersion}~\cite{KV12} is an operational quantity that is closely related to the second-order coding rate. It characterizes the speed at which the rate of optimal   codes converge to the first-order fundamental limit. For  conditional rate-distortion, we may define the dispersion $V_{\mathsf{dps}}$ as
\begin{align}
V_{\mathsf{dps}} \triangleq \lim_{\epsilon \to 0} \limsup_{n\to\infty } \left( \frac{\sqrt{n } (R(\epsilon, D, n) -  R(X;D|S) )}{Q^{-1}(\epsilon)} \right)^2.
\end{align}
From Theorem \ref{C5T1}, we   observe that the operational quantity $V_{\mathsf{dps}}$  is equal to the information quantity $V$. 

Let $V_s\triangleq   \mathsf{var} ( j_{X|S}(X,D|S) \, |\,  S=s)$ be the dispersion\footnote{Note that term {\em dispersion}~\cite{KV12}  here refers to the unconditional rate-distortion problem. This should not cause any confusion in the sequel.}  of the source $X_s\sim P_{X|S}(\cdot|s)$. Now notice that by the law of total variance, $V$ can be decomposed as 
\begin{align}
V &= \mathbb{E} \big[  \mathsf{var} ( j_{X|S}(X,D|S) \, |\,  S) \big]+ \mathsf{var} \big[ \mathbb{E}  ( j_{X|S}(X,D|S ) \, |\, S ) \big] \\
&=\mathbb{E} [V_S] + \mathsf{var} [ R(P_{X|S} ( \cdot |S), D ) ] . 
\end{align}
The first term represents the randomness of the source  weighted by the probability mass function of the side information, while the second term represents the randomness of the side information in terms of the constituent rate-distortion functions. 
 
Theorem \ref{C5T1} is proved in subsection \ref{C5ProofT1}. One of the key ideas in the achievability proof of Theorem \ref{C5T1} is to apply  the random coding bound (Lemma \ref{C5Lrcb}) in the asymptotic evaluation. The key idea in the converse proof of Theorem \ref{C5T1} is to make use of the non-asymptotic converse bound (Lemma \ref{C5Lconverse}) in the asymptotic evaluation. 

We illustrate this theorem through an example.
\begin{myexample}
Consider the case when the source alphabet $\mathcal{X}$, the reconstruction alphabet $\mathcal{Y}$ and the side information alphabet $S$ are binary $\{0,1\}$. The distortion function is the Hamming distance function $d(x,y) = 1\{x\ne y\}$. Assume $P_S( 1) =a$, $P_S( 0) =1-a$,  $P_X( 1) =b$ and $P_X( 0) =1- b$, for $ 0 < a,b < 1$. Assume $P_{X|S}(1|0) = P_{X|S}( 1| 1) = c$, and $P_{X|S}( 0| 0) = P_{X|S}( 0| 1) =1- c$, for $0 <c < \frac{1}{2}$. It can shown that
\begin{align}
j_{X|S}(x,D|s) &= i_{X|S} (x|s) - H(D)
\end{align}
if $0 < D < c$, and $0$ if $D \geq c$. Note that the conditional $D$-tilted information density in this case is independent of the marginal distributions $P_X$ and $P_S$. Next, we have
\begin{align}
R(X;D|S) &= H(X|S) - H(D) \\
         &= H(c) - H(D)   
\end{align}
if $0 < D < c$, and $0$ if $D \geq c$. 
Here $H(D)$ is the entropy of a Bernoulli($D$) source. 

In this example, we can   show that 
\begin{equation}
V = c(1-c) \log^2 \frac{1-c}{c},
\end{equation}
which is simply the dispersion of a Bernoulli($c$) source.

In general, we have the following corollary.
\begin{mycorollary} \label{C5Cor1}
The second-order rate-distortion function $L^*(\epsilon,D,R(X;D|S))$ for the binary source with binary side information and Hamming distortion function is given by
\begin{align}
L^*(\epsilon,D,R(X;D|S)) = \sqrt{\mathsf{var} [i_{X|S}(X|S)]} Q^{-1}(\epsilon).
\end{align}
\end{mycorollary}
\end{myexample}

\subsection{Remarks concerning Theorem \ref{C5T1} }
 
\begin{enumerate}
\item In fact, it is also straightforward to characterize  $L^*(\epsilon,D,\kappa)$ when $\kappa\ne R(X;D|S)$. We have 
\begin{equation}
L^*(\epsilon,D,\kappa) = \left\{ \begin{array}{cc}
+\infty  &\kappa < R(X;D|S)\\
\sqrt{V} Q^{-1}(\epsilon) & \kappa = R(X;D|S)\\
-\infty &\kappa > R(X;D|S)
\end{array} \right.
\end{equation}
%For all $\kappa > R(X;D|S)$, we have $L^*(\epsilon,D,\kappa) =-\infty$ and for all  $\kappa < R(X;D|S)$, we have $L^*(\epsilon,D,\kappa) = + \infty$. 
The first statement above (for the case $\kappa< R(X;D|S)$)  implies the strong converse for conditional rate-distortion. The  strong converse for  unconditional rate-distortion  for discrete memoryless sources is already well known (e.g., \cite[Chapter~7]{CK2011}). 
\item From Theorem \ref{C5T1}, we can deduce that there exists a sequence of  $(M_n,n,D,\epsilon_n)$-codes for the source coding system with side information such that its rate is 
\begin{align}
\frac{1}{n} \log M_n = R(X;D|S) + \sqrt{\frac{V}{n}} Q^{-1}(\epsilon ) + o \left(\frac{1}{\sqrt{n}}\right)
\end{align}
and its asymptotic probability of excess distortion satisfies
\begin{equation}
\epsilon_n\le\epsilon+o(1). 
\end{equation}
It is observed that $V$ characterizes the rate of convergence to the first-order rate-distortion function $R(X;D|S)$.

\item In order to compute $V$, it is noted  that the gradient of $R(X;D|S)$ plays an important role.
\begin{mydef} \label{C5D1}
For each $a \in \mathcal{X}, b \in \mathcal{S}$, define
\begin{align}
R'(P_{X|S}(a|b),D|P_S(b)) \triangleq \frac{dR(P_{\bar{X}|\bar{S}},D|P_{\bar{S}})}{dP_{\bar{X}\bar{S}}(ab)}\Big|_{P_{\bar{X}\bar{S}}=P_{XS}}.
\end{align}
The function $R(P_{\bar{X}|\bar{S}},D|P_{\bar{S}})$ can be thought of as that of $|\mathcal{X}||\mathcal{S}|$ variables.  By stacking up $|\mathcal{X}||\mathcal{S}|$ partial derivatives as defined in Definition \ref{C5D1}, we form  the gradient $\nabla R(P_{XS})$ of $R(P_{\bar{X}|\bar{S}},D|P_{\bar{S}})$ evaluated at $P_{XS}$. The joint distribution $P_{XS}$ can be regarded as a length-$|\mathcal{X}||\mathcal{S}|$ vector that sums to one.
\end{mydef}

Even though the conditional $D$-tilted information density $j_{X|S}(X,D|S)$ is useful in characterizing the second-order rate-distortion function, it is not easy to compute. The task of computing $V$ is made easier by the following lemma. 
\begin{mylemma} \label{C5Lconnection}
For any $a \in \mathcal{X}$ and $b \in \mathcal{S}$, we have
\begin{align}
 j_{X|S}(a,D|b) = R'(P_{X|S}(a|b),D|P_S(b)). \label{C5E49}
\end{align}
\end{mylemma}
\begin{proof}
We have
\begin{align}
R'(P_{X|S}(a|b),D|P_S(b)) &=\frac{dR(P_{\bar{X}|\bar{S}},D|P_{\bar{S}})}{dP_{\bar{X}\bar{S}}(ab)}\Big|_{P_{\bar{X}\bar{S}}=P_{XS}} \\
       &= \frac{d\mathbb{E}[j_{\bar{X}|\bar{S}}(\bar{X},D|\bar{S})]}{dP_{\bar{X}\bar{S}}(ab)} \Big|_{P_{\bar{X}\bar{S}}=P_{XS}} \\
       &= \frac{ d\big[\sum_{x,s} P_{\bar{X} \bar{S}}(xs)j_{\bar{X}|\bar{S}}(x,D|s)\big]}{dP_{\bar{X}\bar{S}}(ab)}\Big|_{P_{\bar{X}\bar{S}}=P_{XS}} \\
               &= j_{\bar{X}|\bar{S}}(a,D|b) + \frac{d\mathbb{E}[j_{\bar{X}|\bar{S}}({X},D|{S})]}{dP_{\bar{X}\bar{S}}(ab)} \Big|_{P_{\bar{X}\bar{S}}=P_{XS}}.
\end{align}
Using part 1) of Lemma \ref{C5Ldtid},  it is evident that
\begin{align}
\frac{d\mathbb{E}[j_{\bar{X}|\bar{S}}({X},D|{S})]}{dP_{\bar{X}\bar{S}}(ab)} \Big|_{P_{\bar{X}\bar{S}}=P_{XS}} =0. 
\end{align}
This completes  the proof of the lemma.
\end{proof}
Let us remark that according to~\cite[Theorem 2.2]{KostinaThesis}, the $D$-tilted information density for the source coding without side information is given by
\begin{align}
 j_{X}(a,D) = R'(P_{X}(a),D) - \log e = \frac{dR(P_{\bar{X}},D)}{dP_{\bar{X}}(a)}\Big|_{P_{\bar{X}}= P_X} - \log e.\label{C5E54} 
\end{align}
This is because
\begin{align}
\frac{dR(P_{\bar{X}},D)}{dP_{\bar{X}}(a)}\Big|_{P_{\bar{X}}=P_{X}} &= \frac{d\mathbb{E}[j_{\bar{X}}(\bar{X},D)]}{dP_{\bar{X}}(a)} \Big|_{P_{\bar{X}}=P_{X}} \\
               &= j_{\bar{X}}(a,D) + \frac{d\mathbb{E}[j_{\bar{X}}({X},D)]}{dP_{\bar{X}}(a)} \Big|_{P_{\bar{X}}=P_{X}},
\end{align}
and in this case we have 
\begin{align}
 \frac{d\mathbb{E}[j_{\bar{X}}({X},D)]}{dP_{\bar{X}}(a)} \Big|_{P_{\bar{X}}=P_{X}} = - \log e. \label{eqn:loge}
\end{align}
Observe that the term $-\log e$ is present in the no-side information setting~(\ref{C5E54}) but not in  the side information setting~(\ref{C5E49}). This is due to \eqref{eqn:loge}.
%This explains  the difference in the presence of term $\log e$ between the case with side information inand the case without side information in.

As a consequence of Lemma \ref{C5Lconnection}, the variance of the conditional $D$-tilted information $V$, defined in \eqref{eqn:var_1}--\eqref{eqn:var_2}, can be alternatively expressed as the variance of the gradient $\nabla R(P_{XS})$ with respect to $P_{XS}$, i.e.,
\begin{align}
V &= \mathsf{var}(\nabla R(P_{XS})) \\
  &= \sum_{a \in \mathcal{X}}  \sum_{b \in \mathcal{S}} P_{XS}(ab) [R'(P_{X|S}(a|b),D|P_S(b))]^2  
  -  \bigg[ \sum_{a \in \mathcal{X}}  \sum_{b \in \mathcal{S}} P_{XS}(ab) R'(P_{X|S}(a|b),D|P_S(b)) \bigg]^2.
\end{align}

\item The relationship between the side-information dependent rate-distortion function $R(P_{X|S}(\cdot|s),D)$ and the conditional rate-distortion function $R(P_{X|S},D|P_S)$ is given by the following lemma \cite{Gray72}.
\begin{mylemma} \label{C5L3}
We have
\begin{align}
R(P_{X|S},D|P_S) = \inf_{ \{d_s\}_{s\in \mathcal{S}} \in \mathcal{D}} \sum_{s \in \mathcal{S}} P_S(s) R(P_{X|S}(\cdot|s),d_s), 
\end{align}
where the set $\mathcal{D}$ is defined as
\begin{align}
\mathcal{D} = \bigg\{ \{d_s\}_{s\in \mathcal{S}} \,\bigg|\, \sum_{s \in \mathcal{S}} P_S(s) d_s = D, d_s \geq 0 \bigg\}. 
\end{align}
\end{mylemma} 
Intuitively, any achievable code for the conditional rate-distortion problem can be thought of as a combination of sub-codes for sub-channels with the side information $S=s$ and the excess  distortion $d_s$. The total distortion $D$ is the $P_S$-convex combination of the constituent excess distortions $d_s$.  Note that Ingber-Kochman \cite{IK11} used the method of types (similarly to the technique used in  Marton's covering lemma~\cite{Marton74}) to perform a second-order (dispersion) analysis for the rate-distortion problem without side information. We attempted to adapt their technique  for our setting but  it was not straightforward to generalize their method to the conditional rate-distortion problem at hand.  This is because Lemma \ref{C5Lconnection} intuitively suggests to treat $X$ and $S$ {\em jointly} to obtain the second-order rate-distortion function $L^*(\epsilon,D,R(P_{X|S}, D |P_S))$. However, if the method of types is used, the relationship in Lemma \ref{C5L3} restricts us to treat $X$ conditioning on $S=s$ first, in the achievability proof, in order to obtain the first-order term. However, this method leads to a different  (and, in fact, inferior) second-order term. The beauty in the random coding bound in Lemma \ref{C5Lforward} is that it allows us to treat $X$ and $S$ {\em jointly}.

\end{enumerate}

\section{Gaussian memoryless source with i.i.d. side information} \label{C5sec:G}
In this section, we consider the i.i.d. Gaussian source. More specifically, 
\begin{align}
X_i \sim \mathcal{N} (0,\sigma^2_X).
\end{align}
The side information is given by
\begin{align}
S_i = X_i + Z_i
\end{align}
where $i=1,2,...,n$, 
\begin{align}
Z_i \sim \mathcal{N}(0, \sigma_Z^2)
\end{align}
and $Z_i$ is independent of $X_i$. We consider the squared-error distortion function, i.e., \begin{align}
d(x^n,y^n) \triangleq \sum_{i=1}^n (x_i -y_i)^2.
\end{align}
%It is assumed that
%\begin{align}
%(X_i,S_i) \sim N(\mathbf{0}, A),
%\end{align}
%where $i=1,2,...,n$ and the correlation matrix $A$ is given by
%\begin{align}
%A \triangleq \begin{bmatrix}
%&\sigma_X^2  &\rho^2 \sigma_X \sigma_S \\
%&\rho^2 \sigma_X \sigma_S &\sigma_S^2 
%\end{bmatrix}.
%\end{align}
Define the conditional variance as 
\begin{align}
\sigma_{X|S}^2 \triangleq \frac{\sigma_X^2 \sigma_Z^2}{\sigma_X^2 + \sigma_Z^2}
\end{align}
The case where $D \geq \sigma^2_{X|S}$ is trivial as $R(X;D|S) =0$. It is assumed that $0 < D < \sigma^2_{X|S}$.  
In this case, it is well-known that \cite{Gray72} the conditional rate-distortion function is given by
\begin{align}
R(X;D|S) = \frac{1}{2} \log \frac{\sigma_{X|S}^2 }{D}.
\end{align}

The second-order rate-distortion function in this case is given by the following theorem.
\begin{mythm} \label{C5T2}
The second-order rate-distortion function $L^*(\epsilon,D,R(X;D|S))$ for   Gaussian source coding with side information is given by
\begin{align}
L^*(\epsilon,D,R(X;D|S)) = \sqrt{\frac{1}{2}} Q^{-1}(\epsilon) \log e.
\end{align}
\end{mythm}

This theorem is proved in subsection \ref{C5ProofT2}.

\subsection{Remarks concerning Theorem \ref{C5T2}}
\begin{enumerate}
\item From Theorem \ref{C5T2}, we observe that the dispersion for Gaussian source coding with side information is  $1/2$ nats squared per source symbol. In other words, the second-order rate-distortion function for   Gaussian source coding with side information is the same as that for Gaussian source coding without side information \cite{KV12} even though the rate-distortion functions for both coding problems are different in general. The presence of side information at both the encoder and the decoder does not affect the second-order coding rate. Intuitively, given the side information $s^n$, the encoder and the decoder can adapt to it  and design a second-order optimal sub-code for each source-encoding sub-test channel (indexed by $s^n$). The second-order coding rate for each sub-test channel is basically the same as that for the source coding system without side information. The second-order rate-distortion function for   Gaussian source coding with side information is the average of all second-order coding rates for sub-test channels, when the average is taken with respect to the side information random variable. Thus, this explains the observation.
\item It would be interesting to investigate if the statement mentioned in the previous item still holds when the side information is  available at either only the decoder or only the encoder. Of course, the rate-distortion functions for the cases where the side information is known at both terminals and at the decoder only are identical in the Gaussian case \cite[Chapter 11]{elgamal}. Thus one wonders whether the dispersion remains at $1/2$ nats$^2$ per source symbol for the Gaussian Wyner-Ziv problem~\cite{WZ76}.
\item Scarlett  \cite{Scarlett14} showed that the dispersion for dirty paper coding (Gaussian Gel'fand-Pinsker) is the same as that when there is no interference. Furthermore, he showed that the same holds true even if the interference is not Gaussian but satisfies some mild concentration conditions. It would be interesting to investigate if the same is true in the lossy compression with  (encoder and decoder)  side information scenario.
\end{enumerate}

\section{Markov source with Markov side information} \label{C5sec:M}
So far, we have considered only memoryless sources. In this section, we consider the system in which the source and side information jointly forms an irreducible, ergodic and time-homogeneous Markov chain, i.e., 
\begin{align}
X_1 S_1 \to X_2 S_2 \to \ldots \to X_n S_n.
\end{align}
We further assume that the source alphabet $\mathcal{X}$ and the side information alphabet $\mathcal{S}$ are both finite.
 Denote the stationary distribution of this Markov chain as $\pi_{XS}$. Assume that this Markov chain starts from the stationary distribution, i.e.,
\begin{align} 
P_{X_1 S_1} = \pi_{XS}. \label{eqn:starting}
\end{align}
Under the assumption in  \eqref{eqn:starting}, all the marginals $P_{X_i S_i}$ for $i\ge 1$ are equal to $\pi_{XS}$. 

First, we define a few relevant quantities.
\begin{mydef}
Define
\begin{align}
\mu &\triangleq  R(X;D|S)\big|_{P_{XS} = \pi_{XS}}, \\
V_n  &\triangleq \mathsf{var} \left(\frac{1}{\sqrt{n}} \sum_{i=1}^n j_{X_i|S_i}(X_i,D|S_i)\right).
\end{align}
\end{mydef}

We have the following important lemma.
\begin{mylemma}
For the Markov chains considered above, the following limit exists
\begin{align}
\lim_{n \to \infty} V_n 
\end{align}
and is equal to 
\begin{align}
V_{\infty} \triangleq \mathsf{var}[j_{X|S}(X,D|S)]\big|_{P_{XS} = \pi_{XS}} + 2\sum_{i=1}^{\infty} \mathsf{cov}[ j_{X_1|S_1}(X_1,D|S_1),j_{X_{1+i}|S_{1+i}}(X_{1+i},D|S_{1+i})]. \label{eqn:V-inf}
\end{align} 
\end{mylemma}
\begin{proof}
The lemma follows from the fact that
\begin{align}
V_n  &= \frac{1}{n}  \mathsf{var} \left(  \sum_{i=1}^n j_{X_i|S_i}(X_i,D|S_i) \right) = \frac{1}{n} \sum_{k,l=1}^n  \mathsf{cov} \left[  j_{X_k|S_k}(X_k,D|S_k),j_{X_l|S_l}(X_l,D|S_l) \right] \\ 
&= \mathsf{var}[j(X,D|S)]\big|_{P_{XS} = \pi_{XS}} + \frac{2}{n} \sum_{j =1}^n (n-j) \mathsf{cov} \left[  j_{X_1|S_1}(X_1,D|S_1),j_{X_{1+j}|S_{1+j}}(X_{1+j},D|S_{1+j}) \right] .  \label{eqn:v_n} %\\ 
%\mathsf{var}[j(X,D|S)]\big|_{P_{XS} = \pi_{XS}} +  \frac{1}{n} \sum_{i=1}^{n} (n -i) \mathsf{cov}[j_{X_1|S_1}(X_1,D|S_1),j_{X_{1+i}|S_{1+i}}(X_{1+i},D|S_{1+i})]. \label{eqn:v_n}
\end{align}
The equality in \eqref{eqn:v_n}  follows from the time-homogeneity of the chain and simple rearrangements.
Now, since the covariance $|\mathsf{cov} \left(  j_{X_1|S_1}(X_1,D|S_1),j_{X_{1+j}|S_{1+j}}(X_{1+j},D|S_{1+j}) \right)|$ decays exponentially fast in the lag $j$ for this class of Markov chains, 
\begin{equation}
\lim_{n\to\infty}\sum_{j =1}^n j \cdot   \mathsf{cov} \left[  j_{X_1|S_1}(X_1,D|S_1),j_{X_{1+j}|S_{1+j}}(X_{1+j},D|S_{1+j}) \right] =0,
\end{equation}
and thus
\begin{align}
\lim_{n\to\infty} V_n = \mathsf{var}[j(X,D|S)]\big|_{P_{XS} = \pi_{XS}} + 2 \sum_{j =1}^{\infty} \mathsf{cov} \left[  j_{X_1|S_1}(X_1,D|S_1),j_{X_{1+j}|S_{1+j}}(X_{1+j},D|S_{1+j}) \right].
\end{align}
The right-hand-side is exactly $V_\infty $ as desired.
\end{proof}

The second-order rate-distortion function for the Markov sequence is given by the following theorem.
\begin{mythm} \label{C5T3}
The second-order rate-distortion function $L^*(\epsilon,D,\mu)$ for the Markov source with side information is given by
\begin{align}
L^*(\epsilon,D,\mu) = \sqrt{V_{\infty}} Q^{-1}(\epsilon).
\end{align}
\end{mythm}

This theorem is proved in subsection \ref{C5ProofT3} and it uses a Markov generalization of the Berry-Ess\'een theorem due to Tikhomirov~\cite{Tikhomirov}.

\subsection{Remarks concerning Theorem \ref{C5T3}}
\begin{enumerate}
\item Notice that the second-order coding rate for the Markov case consists of two parts: 
\begin{equation}
A \triangleq \mathsf{var}[j_{X|S}(X,D|S)]\big|_{P_{XS} = \pi_{XS}},\quad\mbox{and}\quad B \triangleq \sum_{i=1}^{\infty} \mathsf{cov}[j_{X_1|S_1}(X_1,D|S_1),j_{X_{1+i}|S_{1+i}}(X_{1+i},D|S_{1+i})]. \label{eqn:twoparts}
\end{equation}
  When the sequence of random variables $\{X_i S_i\}_{i=1}^{\infty}$ is independent  and identically distributed, the second part $B$ in \eqref{eqn:twoparts} vanishes and we recover the result in section \ref{C5sec:DM}. Thus, the infinite sum in the definition of $V_\infty$ in \eqref{eqn:V-inf} quantifies the effect that the mixing of the Markov chain $\{X_i S_i\}_{i=1}^{\infty}$ has on rate of convergence the finite blockength rate-distortion function to the Shannon limit. The faster the mixing is, the faster the convergence to the Shannon limit is. 

%\item Note that every  sequence  of independently and identically distributed random variables is a special type of Markov chains. In fact if the sequence $\{X_i\}_{i=1}^{\infty}$ is independently and identically distributed, the sequence $S_1 \to S_2 \to S_3 \to \ldots$ forms a Markov chain, and $\{X_i\}_{i=1}^{\infty}$ is jointly independent of $\{S_i\}_{i=1}^{\infty}$, then $X_1 S_1 \to X_2 S_2 \to  X_3 S_3 \to \ldots $ forms a Markov chain. Thus, the result in this section also applies to this special case. \\

\item Denote $\Xi$ as transitional matrix of the Markov chain $X_1 S_1 \to X_2 S_2 \to \ldots \to X_n S_n$. If $\Xi$ is diagonalizable, we can compute $V_{\infty}$ using the following lemma.
\begin{mylemma}
Assume $\Xi = U \mathsf{diag}(1, \lambda_2,\ldots,\lambda_{|\mathcal{X}|   |\mathcal{S}|})U^{\dagger}$. We have 
\begin{align}
V_{\infty} = \mathsf{cov}[j_{X|S}(X,D|S), j_{X'|S'}(X',D|S')]\big|_{P_{XS,X'S'} = \pi_{XS}P_{X'S'|XS}}
\end{align}
where 
\begin{align}
P_{X'S'|XS} (x's'|xs) = \left[U \mathsf{diag}\left(1,\frac{1+ \lambda_2}{1- \lambda_2},\ldots,\frac{1+ \lambda_{|\mathcal{X}|   |\mathcal{S}|}}{1- \lambda_{|\mathcal{X}|   |\mathcal{S}|}} \right) U^{\dagger} \right]_{x's'xs}.
\end{align}
\end{mylemma}
This lemma can be proved using techniques presented by Tomamichel and Tan in \cite[Appendix A]{TomTan13b}. Briefly, we make use of the fact that the Markov chain $\{X_i S_i\}_{i=1}^{\infty}$ is time-homogeneous and starts from the stationary distribution. Secondly, in the diagonalization of the transition matrix $\Xi$, except for eigenvalue $\lambda_1 \triangleq 1$, the rest of the eigenvalues satisfy  $|\lambda_i| <1$. Thus, we have $\sum_{k=1}^{\infty} \lambda_i^k = \frac{\lambda_i}{1-\lambda_i}$ for all but the leading eigenvalue.
\end{enumerate}
\section{Conclusion}
In this paper, the second-order coding rates for the source coding problem with side information available at both the encoder and the decoder are characterized for three different kinds of sources: discrete   memoryless sources, Gaussian memoryless sources and Markov sources. The conditional $D$-tilted information density is found to play a key role in our second-order analysis. 

One of the interesting findings from our work is that the second order rate-distortion functions are same for both Gaussian source coding without side information and with side information (at the enocder and decoder).   The means that the dispersion for both problems is the same and equal to $1/2$ nats$^2$ per source symbol.  An intriguing open problem emanating from this work is whether the dispersion of the Gaussian Wyner-Ziv system \cite{WZ76} is also $1/2$ nats$^2$ per source symbol. 
\section{Appendix} \label{C5sec:app}
\subsection{Proof of Lemma \ref{C5Lforward}} \label{C5ProofLforward}
Lemma \ref{C5Lforward} is a corollary of Lemma \ref{C5Lrcb}. From Lemma \ref{C5Lrcb}, we can show the existence of  an $(M_n, D, n, \epsilon_n)$-code such that
\begin{align}
\epsilon_n &\leq \mathbb{E} \{\mathbb{E}[(1 - P_{\bar{Y}^n|S^n} (B_D(x^n))|S^n)^M] \} \\
           &= \sum_{s^n} P_{S^n}(s^n)  \mathbb{E}[(1 - P_{\bar{Y}^n|S^n} (B_D(x^n))|S^n = s^n)^M].
\end{align}

Using techniques from \cite[Corollary 2.20]{KostinaThesis}, we can show that for every $s^n$, 
\begin{align}
 &\mathbb{E}[(1 - P_{\bar{Y}^n|S^n} (B_D(x^n))|S^n = s^n)^M]  \nonumber\\*
  &\leq \Pr [j_{X^n|S^n}(X^n,D|S^n) > \log \gamma_n - \log \beta_n - \lambda_n^* \delta_n|S^n=s^n] \notag\\*
           &\quad + \mathbb{E}[|1 - \beta_n \Pr[D -\delta_n \leq d(X^n,Y^{n*}) \leq D|X^n]|^+|S^n=s^n] \notag \\*
           &\quad + e^{-\frac{M}{\gamma_n}} \mathbb{E}[\min(1,\gamma_n \exp (-j_{X^n|S^n} (X^n,D,S^n)))|S^n=s^n], \label{C5ELfw}
\end{align}
for any $\gamma_n, \beta_n$, and $\delta_n$.

Taking the average of both sides of inequality (\ref{C5ELfw}) over all sequences $s^n$ completes the proof of this lemma. 

\subsection{Proof of Theorem \ref{C5T1}} \label{C5ProofT1}
In this subsection, we prove Theorem \ref{C5T1}. The proof makes use of the Berry-Ess\'een Theorem \cite[Theorem~2, Chapter XVI.~5]{FellerII}. This theorem is stated as follows.

\begin{mythm}[Berry-Ess\'een Theorem] \label{LemmaBEscalar}
Let $X_k$, for $k=1,2,\ldots,n$ be independent random variables with $\mu_k = \mathbb{E} [X_k]$, $\sigma_k^2 = \mathsf{var} [X_k]$, $t_k = \mathbb{E} [|X_k - \mu_k|^3]$, $\sigma^2 = \sum_{k=1}^n \sigma_k^2$, and $T= \sum_{k=1}^n t_k$. Then for any $\lambda \in\mathbb{R}$, we have
\begin{align}
\left| \Pr \left[ \sum_{k=1}^n (X_k - \mu_k) \geq \lambda \sigma \right] - Q(\lambda) \right| \leq \frac{6T}{\sigma^3}.
\end{align}
\end{mythm}

\subsubsection{Achievability proof of Theorem \ref{C5T1}}
In this part, we prove that, for any $\delta >0$, $\sqrt{V} Q^{-1}(\epsilon) + \delta$ is second-order $(\epsilon,D,\kappa)$-achievable when $\kappa = R(X;D|S)$.

We apply Lemma \ref{C5Lforward} to construct a sequence of $(M_n, D,n,\epsilon_n)$-codes as follows. Choose $\delta_n = \frac{D}{100}$.

 Similar to the proof in \cite[Lemma 4]{KV12}, it can be proved that
\begin{align}
\Pr[D -\delta_n \leq d(X^n,Y^{n*}) \leq D|X^n=x^n] \geq \frac{C}{\sqrt{n}},
\end{align}
when $n$ is sufficiently large, for some constant $C$. Intuitively, this is because $\mathbb{E}[d(X_i,Y_{i}^*)]$ has mean $D$, finite variance and finite absolute third-order moment. Thus, we can apply Theorem \ref{LemmaBEscalar} here.

Choose $\beta_n = \frac{\sqrt{n}}{C}$. We have
\begin{align}
\mathbb{E}[\mathbb{E}[|1 - \beta_n \Pr[D -\delta_n \leq d(X^n,Y^{n*}) \leq D|X^n]|^+|S^n]] = 0,	
\end{align} 
when $n$ is sufficiently large.

Choose $\gamma_n = \frac{M}{\sqrt{n}}$. We have
\begin{align}
&e^{-\frac{M}{\gamma_n}} \mathbb{E}\{\mathbb{E}[\min(1,\gamma_n \exp (-j_{X^n|S^n} (X^n,D,S^n)))|S^n]\} \notag \\
&=e^{-\sqrt{n}}\mathbb{E}\{\mathbb{E}[\min(1,\gamma_n \exp (-j_{X^n|S^n} (X^n,D,S^n)))|S^n]\}   \\
&\leq e^{-\sqrt{n}} \mathbb{E} \{ \mathbb{E}[1|S^n]\} \\
&=e^{-\sqrt{n}}.        
\end{align}

Choose 
\begin{align}
\log M_n = n R(X;D|S) + \sqrt{nV} Q^{-1} (\hat{\epsilon}_n) + \log \sqrt{n} + \lambda_n^* \frac{D}{100} + \log \frac{\sqrt{n}}{C},
\end{align}
where 
\begin{align}
\hat{\epsilon}_n &\triangleq \epsilon - \frac{B_n}{\sqrt{n}} -  e^{-\sqrt{n}} \\
B_n              &\triangleq 6\frac{T_n}{V^{3/2}}  \label{C5EBn} \\
T_n &\triangleq \frac{1}{n} \sum_{i=1}^{n} \mathbb{E}[|j_{X|S}(X,D|S)- R(X;D|S)|^3].
\end{align}

Applying Lemma  \ref{C5Lforward},  for $n$ sufficiently large, we have
\begin{align}
\epsilon_n &\leq \Pr \left[j_{X^n|S^n} (X^n,D|S^n) > n R(X;D|S) + \sqrt{nV} Q^{-1} (\hat{\epsilon}_n) \right]  + e^{-\sqrt{n}} \\*
           &\leq \Pr \left[\sum_{i=1}^n j_{X|S} (X_i,D|S_i) > n R(X;D|S) + \sqrt{nV} Q^{-1} (\hat{\epsilon}_n) \right]  + e^{-\sqrt{n}} \\*
           &\leq \epsilon \label{C5Eforward} 
\end{align}
where equation (\ref{C5Eforward}) follows from Theorem \ref{LemmaBEscalar}.

Therefore, we have constructed a sequence of $(M_n, D,n,\epsilon_n)$-codes satisfying
\begin{align}
\limsup_{n \to \infty} \frac{1}{\sqrt{n}} (\log M_n - n R(X;D|S)) &=\sqrt{{V}} Q^{-1}(\epsilon) \\ 
                                                            % &\leq  \sqrt{{V}} Q^{-1}(\epsilon) + \sigma,\\
                           \limsup_{n \to \infty} \epsilon_n &\leq \epsilon.
\end{align}

\subsubsection{Converse proof of Theorem \ref{C5T1}}
Let $L$ be a second-order $(\epsilon,D,R(X;D|S))$-achievable. We want to show   $Q^{-1}(\epsilon) \sqrt{V} \leq L + \delta$, for any $\delta >0$.

Since $L$ is second-order $(\epsilon,D,R(X;D|S))$-achievable, by definition, there exists a sequence of   $(M_n,n,D,\epsilon_n)$-codes satisfying
 \begin{align}
 \log M_n  &\leq  n R(X;D|S) + \sqrt{n} (L + \delta), \label{C5Econverse1}\\
 \limsup_{n \to \infty} \epsilon_n  &\leq \epsilon,  \label{C5Econverse4}
 \end{align}
when $n$ is sufficiently large.

Using Lemma \ref{C5Lconverse} for $M_n$ satisfying equation (\ref{C5Econverse1}) and $\gamma = \log \sqrt{n}$, we have
\begin{align}
\epsilon_n &\geq  \Pr[j_{X^n|S^n}(X^n,D|S^n) \geq \log M_n + \log \sqrt{n}] - \frac{1}{\sqrt{n}} \\
           &=    \Pr\bigg[\sum_{i=1}^n j_{X|S}(X_i,D|S_i) \geq \log M_n + \log \sqrt{n}\bigg] - \frac{1}{\sqrt{n}} \\
           &\geq \Pr\bigg[\sum_{i=1}^n j_{X|S}(X_i,D|S_i) \geq  n R(X;D|S) + \sqrt{n} (L + \delta)  + \log \sqrt{n} \bigg]   - \frac{1}{\sqrt{n}} \\ 
           &\geq \Pr \bigg[\sum_{i=1}^n j_{X|S}(X_i,D|S_i) - n R(X;D|S) \geq \sqrt{nV} \bigg( \frac{L + \delta}{\sqrt{V}}  + \frac{\log \sqrt{n}}{\sqrt{nV}} \bigg)  \bigg]  - \frac{1}{\sqrt{n}} \\ 
           &\geq Q \bigg( \frac{L + \delta}{\sqrt{V}}  + \frac{\log \sqrt{n}}{\sqrt{nV}} \bigg) - \frac{B_n}{\sqrt{n}}- \frac{1}{\sqrt{n}} \label{C5Econverse2} \\
           &= Q \bigg( \frac{L + \delta}{\sqrt{V}}   \bigg) + O \bigg( \frac{\log \sqrt{n}}{\sqrt{n}} \bigg) - \frac{B_n +1}{\sqrt{n}} \label{C5Econverse3}
\end{align}
where equation (\ref{C5Econverse2}) follows from Theorem  \ref{LemmaBEscalar} and in this equation $B_n$ is defined in (\ref{C5EBn}), and (\ref{C5Econverse3}) follows from the continuity of $Q(\cdot)$ and Taylor expansion.

Combining (\ref{C5Econverse3}) and (\ref{C5Econverse4}), we have
\begin{align}
\epsilon &\geq \limsup_{n \to \infty} \epsilon_n \\*
         &= Q \bigg( \frac{L + \delta}{\sqrt{V}}\bigg). 
\end{align}
Thus, all second-order achievable rates $L$ must satisfy $L\ge Q^{-1}(\epsilon) \sqrt{V}-\delta$. Taking $\delta\downarrow 0$, we complete the proof of the converse.% are second-order achievable.

\subsection{Proof of Theorem \ref{C5T2}} \label{C5ProofT2}
Define the correlation coefficient $\rho$ between $X_i$ and $S_i$, for $i=1,2,\ldots,n$ as
\begin{align}
\rho \triangleq \frac{\mathbb{E}[XS]}{\sqrt{ \mathbb{E}[X^2]\mathbb{E}[S^2] }}=\frac{\sigma_X}{\sqrt{\sigma_Z^2 + \sigma_X^2}}.
\end{align}
Next, we define the conditional mean of $X$ given $S=s$ as
\begin{align}
\mu(s) \triangleq \rho\cdot \frac{\sigma_X}{\sigma_S}\cdot s= \rho^2\cdot  s = \frac{\sigma_X^2}{\sigma_Z^2 + \sigma_X^2}\cdot s.
\end{align}
This is simply the minimum mean squared estimate of $X$ given $S=s$. 

\subsubsection{Achievability proof of Theorem \ref{C5T2}}
In this part, we prove that, for any $\delta >0$, $\sqrt{\frac{1}{2}}  Q^{-1}(\epsilon) \log (e) + \delta$ is second-order $(\epsilon,D,\frac{1}{2}\log \frac{\sigma_{X|S}^2}{D})$-achievable. We apply Lemma \ref{C5Lrcb} to construct a sequence of $(M_n, D,n,\epsilon_n)$-codes as follows.
For each $s^n$, choose  the distribution $P_{\bar{Y}^n|S^n} (\cdot |S^n = s^n)$ in equation (\ref{C5Ercb}) as the uniform distribution on the surface of the $n$-dimensional sphere, with radius $r_0 \triangleq \sqrt{n (\sigma^2_{X|S} -D )}$ and centre at
\begin{align}
\mu(s^n) \triangleq (\mu(s_1), \mu (s_2),\ldots, \mu(s_n)).
\end{align}
Observe that $P_{\bar{Y}^n|S^n} (B_D(x^n))|S^n=s^n) =0$ if
\begin{align}
|x^n - \mu(s^n)| < \sqrt{n (\sigma^2_{X|S} -D )} - \sqrt{nD} \triangleq r_1 
\end{align}
or 
\begin{align}
|x^n - \mu(s^n)| > \sqrt{n (\sigma^2_{X|S} -D )} + \sqrt{nD} \triangleq r_2. 
\end{align}
Therefore, we have a sequence of $(M_n, D,n,\epsilon_n)$-codes that satisfies
\begin{align}
\epsilon_n &\leq \mathbb{E} \{\mathbb{E}[(1 - P_{\bar{Y}^n|S^n} (B_D(X^n))|S^n)^{M_n}] \} \\
           & \leq \mathbb{E} \{\mathbb{E}[(1 - P_{\bar{Y}^n|S^n} (B_D(X^n)))^{M_n} .\Pr(r_1 \leq |x^n - \mu(S^n)| \leq r_2)|S^n] \} \nonumber \\*
           &\quad +   \mathbb{E} \{\mathbb{E}[\Pr(r_2 < |X^n - \mu(S^n)|)|S^n] \} \nonumber \\*
           &\quad +  \mathbb{E} \{\mathbb{E}[\Pr(r_1 > |X^n - \mu(S^n)|)|S^n] \}.
\end{align}
By the weak law of large numbers, we observe that the second term and the third term become vanishingly small as $n \to \infty$. Now, we  analyze the first term. 

\begin{figure}
\centering
\includegraphics[width=0.70\textwidth]{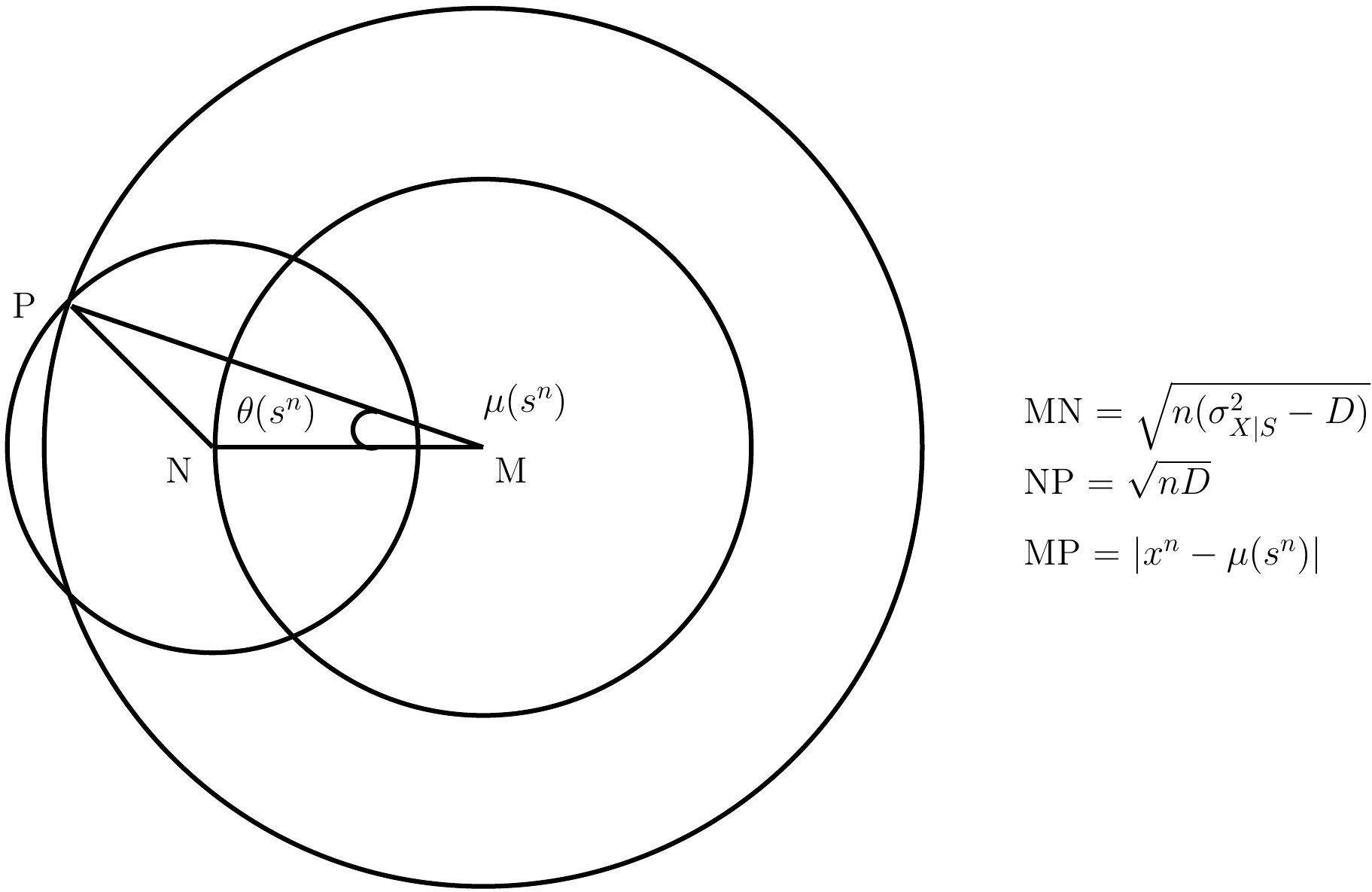}
\caption{Encoding for Gaussian source}
\label{fig:C5Gaussian}
\end{figure}

Note that $\frac{|X^n - \mu(s^n)|^2}{\sigma^2_{X|S}}$ has a central $\chi_n^2$ distribution. Denote $A_n(r_0) \triangleq \frac{n\pi^{\frac{n}{2}}}{\Gamma(\frac{n}{2} +1)} r_0^{n-1}$ as the surface area of an $n$-dimensional sphere of radius $r_0$. Denote $A_n(r_0, \theta(s^n))$ as the surface area of $n$-dimensional polar cap of radius $r_0$ and angle $\theta(s^n)$ (see Figure \ref{fig:C5Gaussian}), where the angle $0< \theta(s^n) < \pi$ is given by
\begin{align}
\theta(s^n) \triangleq \cos^{-1} \left(\frac{|x^n - \mu(s^n)|^2 + r_0^2 - nD}{2|x^n - \mu(s^n)|r_0} \right).
\end{align}
We have
\begin{align}
& \mathbb{E} \{\mathbb{E}[(1 - P_{\bar{Y}^n|S^n} (B_D(X^n)))^{M_n} .\Pr(r_1 \leq |x^n - \mu(S^n)| \leq r_2)|S^n] \} \\
  &= \mathbb{E} \left\{\mathbb{E} \left[ \left(1 - \frac{A_n(r_0)}{A_n(r_0, \theta(s^n))} \right)^{M_n} .\Pr(r_1 \leq |x^n - \mu(S^n)| \leq r_2)\Big|S^n \right] \right\} \label{C5EG1} \\
  &\leq \mathbb{E} \left\{\mathbb{E} \left[ \left(1 - \frac{\Gamma(\frac{n}{2} +1)}{\sqrt{\pi}n \Gamma(\frac{n-1}{2}+1)} (\sin (\theta(s^n)))^{n-1} \right)^{M_n} .\Pr(r_1 \leq |x^n - \mu(S^n)| \leq r_2)\Big|S^n \right] \right\} \label{C5EG2} \\
  &\leq \mathbb{E} \left\{  \left[ n \int_{0}^{\infty} \left( 1 - f(n,z) \right)^{M_n} 1\{ r_1 \leq z \leq r_2 \} P_{\chi_n^2}(nz)dz   \Big|S^n \right] \right\} \label{C5EG3} \\
  &= n \int_{0}^{\infty} \left( 1 - f(n,z) \right)^{M_n} 1\{ r_1 \leq z \leq r_2 \} P_{\chi_n^2}(nz)dz  \label{eqn:thm37_kv}   
\end{align}
where 
\begin{itemize}
\item (\ref{C5EG1}) comes from geometry,
\item (\ref{C5EG2}) comes from a lower bound on $A_n(r_0, \theta(s^n))$ \cite{Sakrison}, and
\item  in (\ref{C5EG3}),  the function $f(n,z)$ is defined as
\begin{align}
f(n,z) \triangleq \frac{\Gamma(\frac{n}{2} +1)}{\sqrt{\pi}n \Gamma(\frac{n-1}{2}+1)} \left(1 - \frac{\left(1+ z - 2\frac{D}{\sigma^2_{X|S}}\right)^2}{4 \left(1- \frac{D}{\sigma^2_{X|S}} \right)z}   \right)^{\frac{n-1}{2}}
\end{align}
and $P_{\chi_n^2}$ is  the central $\chi_n^2$ probability density function.
\end{itemize}  
Next, we choose the sequence $M_n$ such that
\begin{align}
\frac{\log M_n}{n} = \frac{1}{2} \log\frac{\sigma^2_{X|S}}{D} + \sqrt{\frac{1}{2n}}Q^{-1}(\epsilon) \log e +  \frac{\log n}{2n} + \frac{\log \log n}{n} + O \left(\frac{1}{n} \right).
\end{align}
We can check that
\begin{align}
\limsup_{n \to \infty} \frac{1}{\sqrt{n}} \left(\log M_n -  \frac{n}{2} \log\frac{\sigma^2_{X|S}}{D} \right) =\sqrt{\frac{1}{2}}Q^{-1}(\epsilon) \log e.
\end{align}                                                                  
Using similar techniques as in \cite[Appendix K]{KV12}, we can show that the bound in \eqref{eqn:thm37_kv} can be analyzed using the Gaussian approximation to yield
\begin{align}
 \limsup_{n \to \infty} \epsilon_n &\leq \epsilon.
\end{align}

\subsubsection{Converse proof of Theorem \ref{C5T2}}
The conditional $D$-tilted information in the jointly Gaussian case is
\begin{align}
j_{X^n|S^n}(x^n,D|s^n) = \frac{n}{2} \log \frac{\sigma_{X|S}^2}{D} + \frac{\left|x^n -   \frac{\sigma_X^2}{ \sigma_X^2 + \sigma_Z^2}  s^n\right|^2}{2\sigma_{X|S}^2} \log e - \frac{n}{2}\log e.
\end{align}
For each $i \in \{1,2,\ldots,n\}$, we have
\begin{align}
\mathbb{E} [j_{X_i|S_i}(X_i,D|S_i)] = \frac{1}{2}\log \frac{\sigma_{X|S}^2}{D},
\end{align}
and
\begin{align}
\mathsf{var} [j_{X_i|S_i}(X_i,D|S_i)] &= \mathbb{E} \left[\frac{|X_i - \mu(S_i)|^2}{2 \sigma_{X|S}^2}  \log e - \frac{\log e}{2}   \right]^2 \\
                                      &= \mathbb{E}\left[\mathbb{E} \left[\left(\frac{|X_i - \mu(S_i)|^2}{2 \sigma_{X|S}^2}  \log e - \frac{\log e}{2}   \right)^2 \bigg|S_i \right]   \right] \\
                                      &= (\log e)^2 \mathbb{E}\left[\mathbb{E} \left[\left(\frac{|X_i - \mu(S_i)|^4}{4 \sigma_{X|S}^4}   -  \frac{|X_i - \mu(S_i)|^2}{2 \sigma_{X|S}^2} + \frac{1}{4} \right) \bigg|S_i \right]   \right] \\
                                      &= \frac{1}{2}(\log e)^2. 
\end{align}
Let $L$ be  second-order $(\epsilon,D,\frac{1}{2}\log \frac{\sigma_{X|S}^2}{D})$-achievable. We want to show that $Q^{-1}(\epsilon) \sqrt{\frac{1}{2}} \log e \leq L + \delta$, for any $\delta >0$. Since $L$ is second-order $(\epsilon,D,\frac{1}{2}\log \frac{\sigma_{X|S}^2}{D})$-achievable, there exists a sequence of   $(M_n,n,D,\epsilon_n)$-codes satisfying
 \begin{align}
 \log M_n  &\leq    \frac{n}{2}\log \frac{\sigma_{X|S}^2}{D} + \sqrt{n} (L + \delta), \label{C5EconverseG}\\
 \limsup_{n \to \infty} \epsilon_n  &\leq \epsilon,  \label{C5EconverseG4}
 \end{align}
where \eqref{C5EconverseG} holds for  all $n$ sufficiently large.

Using Lemma \ref{C5Lconverse} for $M_n$ satisfying equation (\ref{C5EconverseG}) and $\gamma = \log \sqrt{n}$, we have
\begin{align}
\epsilon_n &\geq  \Pr[j_{X^n|S^n}(X^n,D|S^n) \geq \log M_n + \log \sqrt{n}] - \frac{1}{\sqrt{n}} \\
           &\geq Q \left( \frac{L + \delta}{\sqrt{\frac{1}{2}} \log e}  + \frac{\log \sqrt{n}}{\sqrt{\frac{1}{2}n} \log e} \right) - \frac{B_n}{\sqrt{n}}- \frac{1}{\sqrt{n}} \label{C5EconverseG2} \\
           &= Q \left( \frac{L + \delta}{\sqrt{\frac{1}{2}} \log e}   \right) + O \left( \frac{\log \sqrt{n}}{\sqrt{n}} \right) - \frac{B_n +1}{\sqrt{n}} \label{C5EconverseG3}
\end{align}
where equation (\ref{C5EconverseG2}) follows from  Theorem  \ref{LemmaBEscalar} and in this equation $B_n$ is the constant in  Theorem \ref{LemmaBEscalar}, and (\ref{C5EconverseG3}) follows from the continuous differentiability  of $Q(\cdot)$ and Taylor expansion.

Combining (\ref{C5EconverseG3}) and (\ref{C5EconverseG4}), we have
\begin{align}
\epsilon &\geq \limsup_{n \to \infty} \epsilon_n \\*
         &= Q \left( \frac{L + \delta}{\sqrt{\frac{1}{2}}\log e}\right). 
\end{align}
This completes the proof of the converse upon taking $\delta\downarrow 0$.

\subsection{Proof of Theorem \ref{C5T3}} \label{C5ProofT3}
To prove Theorem \ref{C5T3}, we use a variant of Berry-Ess\'een Theorem \cite{Tikhomirov} to deal with a sequence of random variables that forms a  Markov chain. This theorem is stated as follows.

\begin{mythm} \label{LemmaBEmarkov}
Consider a stationary  process $\{X_k: k\geq 1\}$, with $\mathbb{E}X_1 = 0$ and finite variance. Define the strong mixing coefficient $\alpha(n)$ as
\begin{align}
\alpha(n) \triangleq \sup \{ |\Pr(A \cap B) - \Pr(A)\Pr(B)|:A \in \mathcal{F}_{-\infty}^{k}, B \in \mathcal{F}_{k+n}^{\infty}, k \in \mathbb{Z} \},
\end{align}
where $\mathcal{F}_{a}^{b} = \sigma \left\langle X_i: i \in[a,b]\cap \mathbb{Z} \right\rangle$ is the $\sigma$-field generated by $\{X_i: i \in[a,b]\cap \mathbb{Z}\}$, $-\infty \leq a \leq b \leq \infty$.
 Denote 
\begin{align}
\sigma_n^2 \triangleq \mathbb{E} \left[\left(\sum_{j=1}^n X_j\right)^2 \right].
\end{align}
Assume that the strong mixing coefficient is exponentially decaying, i.e., $\alpha(n) \leq K e^{-\kappa_1 n}$ for some $K$ and $\kappa_1$ and all $n \geq 1$. Assume $\mathbb{E}[|X_1^{2+\gamma}] < \infty$ for some $\gamma$, $1\geq \gamma >0$. 
Then, there is a constant $B(K,\kappa, \gamma) > 0$ such that, for all $n \in \mathbb{N}$,
\begin{align}
\sup_{x \in \mathbb{R}} \left| \Pr \left[ \frac{1}{\sigma_n} \sum_{k=1}^n X_k   \leq \lambda  \right] - \Phi(\lambda) \right| \leq \frac{B(K,\kappa_1, \gamma) (\log n)^{1+ \frac{\gamma}{2}}}{n^{\frac{\gamma}{2}}}.
\end{align}
\end{mythm}
Note that the strong mixing coefficient of a time-homogeneous, irreducible and ergodic Markov chain decays to zero  and, in fact, vanishes exponentially fast\cite[Theorem 3.1]{Bradley}.

In this proof, we make use of the following lemma.
\begin{mylemma} \label{C5LMarkov}
If the sequence $X_1S_1 \to X_2S_2 \to X_3S_3 \to \ldots$ forms a Markov chain, then the sequence of conditionally $D$-tilted information densities $\{j_{X_i|S_i}(X_i,D|S_i)\}_{i=1}^{\infty}$ also forms a Markov chain.
\end{mylemma}
This lemma is proved in section \ref{C5ProofMarkov}
\subsubsection{Achievability proof of Theorem \ref{C5T3}}
In this part, we prove that, for any $\delta >0$, $\sqrt{V_{\infty}} Q^{-1}(\epsilon) + \delta$ is second-order $(\epsilon,D,\mu)$-achievable.

We apply Lemma \ref{C5Lforward} to construct a sequence of $(M_n, D,n,\epsilon_n)$-codes as follows. Choose $\delta_n = \frac{D}{100}$.

 Similar to the proof in \cite[Lemma 4]{KV12}, it can be proved that
\begin{align}
\Pr[D -\delta_n \leq d(X^n,Y^{n*}) \leq D|X^n=x^n] \geq \frac{C}{\sqrt{n}},
\end{align}
when $n$ is sufficiently large, for some constant $C$. Intuitively, this is because $\mathbb{E}[d(X_i,Y_{i}^*)]$ has mean $D$, finite variance. Thus, we can apply Theorem \ref{LemmaBEmarkov} for a sum of weakly dependent variables.

Choose $\beta_n = \frac{\sqrt{n}}{C}$. We have
\begin{align}
\mathbb{E}[\mathbb{E}[|1 - \beta_n \Pr[D -\delta_n \leq d(X^n,Y^{n*}) \leq D|X^n]|^+|S^n]] = 0,	
\end{align} 
when $n$ is sufficiently large.

Choose $\gamma_n = \frac{M}{\sqrt{n}}$. We have
\begin{align}
&e^{-\frac{M}{\gamma_n}} \mathbb{E}\{\mathbb{E}[\min(1,\gamma_n \exp (-j_{X^n|S^n} (X^n,D,S^n)))|S^n]\} \notag \\
&=e^{-\sqrt{n}}\mathbb{E}\{\mathbb{E}[\min(1,\gamma_n \exp (-j_{X^n|S^n} (X^n,D,S^n)))|S^n]\}   \\
&\leq e^{-\sqrt{n}} \mathbb{E} \{ \mathbb{E}[1|S^n]\} \\
&=e^{-\sqrt{n}}.        
\end{align}

Choose 
\begin{align}
\log M_n = n \mu + \sqrt{nV_{n}} Q^{-1} (\hat{\epsilon}_n) + \log \sqrt{n} + \lambda_n^* \frac{D}{100} + \log \frac{\sqrt{n}}{C},
\end{align}
where 
\begin{align}
\hat{\epsilon}_n &\triangleq \epsilon - \frac{B(K,\kappa_1, \gamma)(\log n)^{1+\frac{\gamma}{2}}}{n^{\frac{\gamma}{2}}} -  e^{-\sqrt{n}}
\end{align}
and $B(K,\kappa_1, \gamma)$ is found in Theorem \ref{LemmaBEmarkov}.

Applying Lemma  \ref{C5Lforward},  for $n$ sufficiently large, we have
\begin{align}
\epsilon_n &\leq \Pr \left[j_{X^n|S^n} (X^n,D|S^n) > n \mu  + \sqrt{n V_{n}} Q^{-1} (\hat{\epsilon}_n) \right]  + e^{-\sqrt{n}} \\
           &\leq \Pr \left[\sum_{i=1}^n j_{X_i|S_i} (X_i,D|S_i) > n \mu + \sqrt{nV_n} Q^{-1} (\hat{\epsilon}_n) \right]  + e^{-\sqrt{n}} \\
           &\leq \epsilon \label{C5EforwardM} 
\end{align}
where equation (\ref{C5EforwardM}) follows from   Theorem \ref{LemmaBEmarkov}.

Therefore, we have constructed a sequence of $(M_n, D,n,\epsilon_n)$-codes satisfying
\begin{align}
\limsup_{n \to \infty} \frac{1}{\sqrt{n}} (\log M_n - n \mu ) &=\sqrt{V_{\infty}} Q^{-1}(\epsilon) \\ 
                                                            % &\leq  \sqrt{{V}} Q^{-1}(\epsilon) + \sigma,\\
                           \limsup_{n \to \infty} \epsilon_n &\leq \epsilon.
\end{align}

\subsubsection{Converse proof of Theorem \ref{C5T3}}
Let $L$ be   second-order $(\epsilon,D,\mu)$-achievable. In this part, we want to show that $Q^{-1}(\epsilon) \sqrt{V_{\infty}} \leq L + \delta$, for any $\delta >0$.

Since $L$ is  $(\epsilon,D,\mu)$-second-order achievable  there exists a sequence of   $(M_n,n,D,\epsilon_n)$-codes satisfying
 \begin{align}
 \log M_n  &\leq  n \mu + \sqrt{n} (L + \delta), \label{C5Econverse1M}\\
 \limsup_{n \to \infty} \epsilon_n  &\leq \epsilon,  \label{C5Econverse4M}
 \end{align}
when $n$ is sufficiently large.

Using Lemma \ref{C5Lconverse} for $M_n$ satisfying equation (\ref{C5Econverse1M}) and $\gamma = \log \sqrt{n}$, we have
\begin{align}
\epsilon_n &\geq  \Pr[j_{X^n|S^n}(X^n,D|S^n) \geq \log M_n + \log \sqrt{n}] - \frac{1}{\sqrt{n}} \\*
           &=    \Pr\bigg[\sum_{i=1}^n j_{X_i|S_i}(X_i,D|S_i) \geq \log M_n + \log \sqrt{n}\bigg] - \frac{1}{\sqrt{n}} \\
           &\geq \Pr\bigg[\sum_{i=1}^n j_{X_i|S_i}(X_i,D|S_i) \geq  n \mu + \sqrt{n} (L + \delta)  + \log \sqrt{n} \bigg]  - \frac{1}{\sqrt{n}} \\ 
           &\geq \Pr \bigg[\sum_{i=1}^n j_{X_i|S_i}(X_i,D|S_i) - n \mu \geq \sqrt{nV_n} \bigg( \frac{L + \delta}{\sqrt{V_n}}  + \frac{\log \sqrt{n}}{\sqrt{nV_n}} \bigg)  \bigg]   - \frac{1}{\sqrt{n}} \\ 
           &\geq Q \bigg( \frac{L + \delta}{\sqrt{V_n}}  + \frac{\log \sqrt{n}}{\sqrt{nV_n}} \bigg) - \frac{B(K,\kappa_1,\gamma)(\log n)^{1+\frac{\gamma}{2}}}{n^{\frac{\gamma}{2} }}- \frac{1}{\sqrt{n}} \label{C5Econverse2M} \\*
           &= Q \bigg( \frac{L + \delta}{\sqrt{V_n}}   \bigg) + O \bigg( \frac{\log \sqrt{n}}{\sqrt{n}} \bigg)- \frac{B(K,\kappa_1,\gamma)(\log n)^{1+\frac{\gamma}{2}}}{n^{\frac{\gamma}{2} }}- \frac{1}{\sqrt{n}} \label{C5Econverse3M}
\end{align}
where equation (\ref{C5Econverse2M}) follows from  Theorem  \ref{LemmaBEmarkov} and in this equation $B(K,\kappa_1,\gamma)$ is defined in  Theorem  \ref{LemmaBEmarkov}, and (\ref{C5Econverse3M}) follows from the continuity of $Q(\cdot)$ and Taylor expansion.

Combining (\ref{C5Econverse3M}) and (\ref{C5Econverse4M}), we have
\begin{align}
\epsilon &\geq \limsup_{n \to \infty} \epsilon_n \\
         &= Q \bigg( \frac{L + \delta}{\sqrt{V_{\infty}}}\bigg)  \label{C5Econverse6M}
\end{align}
where in (\ref{C5Econverse6M}), we use the fact that $V_n \to V_{\infty}$.
\subsection{Proof of Lemma \ref{C5LMarkov}} \label{C5ProofMarkov}
In the proof of this lemma, we make use of the following lemma.
\begin{mylemma} \label{C5LemmaTK}
Let $\{A_i\}_{i=1}^{\infty}$ be a Markov chain in state space $\mathcal{A}$. Consider the sequence $\{B_i = f(X_i)\}_{i=1}^{\infty}$, where $f:\mathcal{A} \to \mathcal{B}$ is  a function from $\mathcal{A}$ to  $\mathcal{B}$. Suppose that there   exists a function $g:\mathcal{B} \times \mathcal{B} \to \mathbb{R}$ such that 
\begin{align}
\Pr(B_{i+1} = b|X_i = a) = g(f(a),b)
\end{align}
 for any $a \in \mathcal{A}$ and $b \in \mathcal{B}$. Then the sequence $\{B_i\}_{i=1}^{\infty}$ forms a Markov chain.
\end{mylemma}
The proof of this lemma can be found in \cite[Lemma 13]{Konstan09}. Note that if $f$ is one-to-one, then it is obvious that the sequence generated by $f$ acting on a Markov chain is also a Markov chain.

Here, $j_{X|S}$ is a composition of several functions $\log$, $\frac{1}{t}$ for $t \neq 0$, $\exp$, summation and $d(.|.)$. So, Lemma \ref{C5LMarkov} follows from Lemma \ref{C5LemmaTK}.
\subsection*{Acknowledgments}
The authors would like to thank Anshoo Tandon for several helpful discussions, and also Victoria Kostina for   prompt and detailed clarifications of her works. 

The works of Sy-Quoc Le and Mehul Motani are supported in part by National University of Singapore under Research Grant WBS R-263-000-579-112.

The work of Vincent Tan   is supported by National University of Singapore under Research Grant  R-263-000-A98-750/133.

\bibliographystyle{unsrt}
\bibliography{myrefOPT}

\begin{thebibliography}{10}

\bibitem{Shannon48}
C.~E. Shannon.
\newblock A mathematical theory of communication.
\newblock {\em Bell System Technical Journal}, pages 379--423, 1948.

\bibitem{Shannon59S}
C.~E. Shannon.
\newblock Coding theorems for a discrete source with a fidelity criterion.
\newblock {\em IRE Nat. Conv. Rec.}, pages 142--163, 1959.

\bibitem{Berger71}
T.~Berger.
\newblock {\em Rate Distortion Theory: A Mathematical Basis for Data
  Compression}.
\newblock Prentice-Hall, 1971.

\bibitem{Gray72}
R.~M. Gray.
\newblock Conditional rate-distortion theory.
\newblock {\em Technical Report, Stanford University}, AD-753260, Oct. 1972.

\bibitem{Weissman06}
T.~Weissman and A.~El Gammal.
\newblock Source coding with limited-look-ahead side information at the
  decoder.
\newblock {\em IEEE Transactions on Information Theory}, 52(12):5218--5239,
  Dec. 2006.

\bibitem{WZ76}
A.~D. Wyner and J.~Ziv.
\newblock The rate-distortion function for source coding with side information
  at the decoder.
\newblock {\em IEEE Transactions on Information Theory}, 22(1):1--10, Jan.
  1976.

\bibitem{LG}
B.~M. Leiner and R.~M. Gray.
\newblock Rate-distortion theory for ergodic sources with side information.
\newblock {\em IEEE Transactions on Information Theory}, 20(5):672--675, Sep.
  1974.

\bibitem{FE06}
M.~Fleming and M.~Effros.
\newblock On rate-distortion with mixed types of side information.
\newblock {\em IEEE Transactions on Information Theory}, 52(4):1698--1705, Apr.
  2006.

\bibitem{SP13}
O.~Simeone and H.~H. Permuter.
\newblock Source coding when the side information may be delayed.
\newblock {\em IEEE Transactions on Information Theory}, 59(6):3607--3618,
  June. 2013.

\bibitem{LZZ00}
T.~Linder, R.~Zamir, and K.~Zeger.
\newblock On source coding with side-information-dependent distortion measures.
\newblock {\em IEEE Transactions on Information Theory}, 46(7):2697--2704, Jul.
  2000.

\bibitem{Strassen62}
V.~Strassen.
\newblock Asymptotische abschatzungen in shannon's informationstheorie.
\newblock {\em Trans. Third Prague Conf. Information Theory}, pages 689--723,
  1962.

\bibitem{Hayashi08}
M.~Hayashi.
\newblock Second-order asymptotics in fixed-length source coding and intrinsic
  randomness.
\newblock {\em IEEE Transactions on Information Theory}, 54(10):4619--4637,
  Oct. 2008.

\bibitem{Han_folklore}
T.~S. Han.
\newblock Folklore in source coding: {Information-spectrum} approach.
\newblock {\em IEEE Transactions on Information Theory}, 51(2):747--753, 2005.

\bibitem{KV12}
V.~Kostina and S.~Verd\'u.
\newblock Fixed-length lossy compression in the finite blocklength regime.
\newblock {\em IEEE Transactions on Information Theory}, 58(6):3309--3338, Jun.
  2012.

\bibitem{IK11}
A.~Ingber and Y.~Kochman.
\newblock The dispersion of lossy source coding.
\newblock In {\em Proc. Data Compression Conference}, pages 53--62, 2011.

\bibitem{Tan12}
V.~Y.~F. Tan.
\newblock Moderate-deviations of lossy source coding for discrete and
  {Gaussian} sources.
\newblock In {\em Proc. International Symposium on Information Theory}, pages
  920--924, Cambridge, MA, Jul 2012.

\bibitem{WKT13}
S.~Watanabe, S.~Kuzuoka, and V.~Y.~F. Tan.
\newblock Non-asymptotic and second-order achievability bounds for coding with
  side-information.
\newblock 2013. arXiv:1301.6467.

\bibitem{KontoVerdu12}
I.~Kontoyiannis and S.~Verd\'u.
\newblock Optimal lossless data compression: Non-asymptotics and asymptotics.
\newblock {\em IEEE Transactions on Information Theory}, 60(2):777--795, Feb
  2014.

\bibitem{Tan_mono}
V.~Y.~F. Tan.
\newblock Asymptotic estimates in information theory with non-vanishing error
  probabilities.
\newblock {\em {Foundations and Trends$\,$\textregistered $ $ in Communications
  and Information Theory}}, 11(1--2):1--184, 2014.

\bibitem{YS93}
B.~Yu and T.~P. Speed.
\newblock A rate of convergence result for a universal $d$-semifaithful code.
\newblock {\em IEEE Transactions on Information Theory}, 39(3):813--820, May
  1993.

\bibitem{ZYW}
Z.~Zhang, E.~h.~Yang, and V.~K. Wei.
\newblock The redundancy of source coding with fidelity criterion-part one:
  Known statistics.
\newblock {\em IEEE Transactions on Information Theory}, 43(1):71--91, Jan.
  1997.

\bibitem{KV12CISS}
V.~Kostina and S.~Verd\'u.
\newblock A new converse in rate distortion theory.
\newblock In {\em Proc. Annual Conference on Information Sciences and Systems},
  volume~46, Princeton, NJ, 2012.

\bibitem{CK2011}
I.~Csiszar and J.~Korner.
\newblock {\em Information Theory: Coding Theorems for Discrete Memoryless
  Systems}.
\newblock Cambridge University Press, $2^{nd}$ edition, 2011.

\bibitem{KostinaThesis}
V.~Kostina.
\newblock {\em Lossy Data Compression: Nonasymptotic Fundamental Limits}.
\newblock PhD thesis, Department of Electrical Engineering, Princeton, 2013.

\bibitem{Marton74}
K.~Marton.
\newblock Error exponent for source coding with a fidelity criterion.
\newblock {\em IEEE Transactions on Information Theory}, 20(2):197--199, 1974.

\bibitem{elgamal}
A.~{El~Gamal} and Y.-H. Kim.
\newblock {\em Network Information Theory}.
\newblock Cambridge University Press, Cambridge, U.K., 2012.

\bibitem{Scarlett14}
J.~Scarlett.
\newblock On the dispersion of dirty paper coding.
\newblock In {\em Proc. IEEE International Symposium on Information Theory},
  pages 2282--2286, Honolulu, HI, Jul 2014.
\newblock {\tt arXiv:1309.6200 [cs.IT]}.

\bibitem{Tikhomirov}
A.~N. Tikhomirov.
\newblock On the convergence rate in the central limit theorem for weakly
  dependent random variables.
\newblock {\em Theory of Probability and its Applications}, 25(4):790--809,
  1980.

\bibitem{TomTan13b}
M.~Tomamichel and V.~Y.~F. Tan.
\newblock Second-order coding rates for channels with state.
\newblock {\em IEEE Transactions on Information Theory}, 60(8):4427--4448, Aug.
  2014.

\bibitem{FellerII}
W.~Feller.
\newblock {\em An Introduction to Probability Theory and Its application},
  volume~II.
\newblock John Wiley and Sons, 2nd edition, 1971.

\bibitem{Sakrison}
D.~Sakrison.
\newblock A geometric treatment of the source encoding of a {G}aussian random
  variable.
\newblock {\em IEEE Transactions on Information Theory}, 14(3):481--486, May
  1968.

\bibitem{Bradley}
R.~C. Bradley.
\newblock Basic properties of strong mixing conditions. {A} survey and some
  open questions.
\newblock {\em Probability Surveys}, 2:107--144, 2005.

\bibitem{Konstan09}
T.~Konstantopoulos.
\newblock Introductory lecture notes on {M}arkov chains and random walks.
\newblock 2009. Available at http://www2.math.uu.se/~takis/L/McRw/mcrw.pdf.

\end{thebibliography}
\end{document}